\newif\ifarxiv
\arxivtrue

\ifarxiv
    \documentclass[11pt]{article}
    \usepackage{float}
    \usepackage{fullpage}
    \usepackage[round]{natbib}
    \usepackage{hyperref}
    \usepackage{palatino}
    \usepackage{setspace}
\else
    \documentclass[preprints,article,accept,moreauthors,pdftex]{Definitions/mdpi}
\fi

\usepackage[export]{adjustbox}
\usepackage{amsfonts}
\usepackage{amsmath}
\usepackage{amssymb}
\usepackage{amsthm}
\usepackage{authblk}
\usepackage{blkarray}
\usepackage{booktabs}
\usepackage{braket}
\usepackage{todonotes}
\usepackage{tikz}
\usetikzlibrary{quantikz}

\newcommand{\F}{\mathbb{F}_2}
\newcommand{\cnot}{\mathrm{CNOT}}
\newtheorem{thm}{Theorem}
\theoremstyle{definition}

\newtheorem{lem}[thm]{Lemma}
\theoremstyle{definition}

\newtheorem{defn}{Definition}
\theoremstyle{definition}

\newtheorem{eg}{Example}

\newcommand\abstracttext{This paper presents a system for solving binary-valued linear equations using quantum computers.
The system is called Mod2VQLS, which stands for Modulo2 Variational Quantum Linear Solver.
As far as we know, this is the first such proposal. 
The design is a classical-quantum hybrid. The quantum components are a new circuit design for 
implementing matrix multiplication modulo 2, and a variational circuit to be optimized. The classical components
are the optimizer which measures the cost function and updates the quantum parameters for each iteration, and 
the controller that runs the quantum job and classical optimizer iterations.
We propose two alternative ansatze or templates for the variational circuit, and present results showing
that the rotation ansatz designed specifically for this problem provides the most direct path to a valid solution. Numerical experiments in low dimensions indicate that Mod2VQLS, using the
custom rotations ansatz, is on-a-par with the block Wiedemann algorithm,
the best-known to date for this problem.
}

\ifarxiv \else 
\input{mdpi_header.tex} 
\abstract{\abstracttext} 
\fi

\begin{document} 

\ifarxiv
    \title{Mod2VQLS: a Variational Quantum Algorithm for Solving Systems of Linear Equations Modulo 2}
    \author{Willie Aboumrad and Dominic Widdows}
    \affil{IonQ, Inc.} 
    \date{\today}
    \maketitle
    \abstract{\abstracttext}
\fi

\setlength{\headheight}{17.38797pt}

\section{Introduction}

This report describes a new quantum hybrid algorithm for solving systems of linear equations modulo $2$.
In such problems, we want to find an $n$-vector $x$ such that $Ax = b$, where $A$ is an $m \times n$ matrix and $b$ is an $m$-vector, all values are $0$ or $1$, and arithmetic is performed modulo $2$, so that $1 + 1 = 0$.

Solving systems of simultaneous linear equations is a standard
problem, and is most familiar in situations where the coefficients in the equations are real or complex numbers. When the coefficient matrix is not square, direct methods leveraging the LU or QR decompositions emerge as the natural choice. Obtaining these factorizations is costly, incurring a computational cost of $O(m n^2)$ \cite{trefethen1997numerical}. When $A$ is square, and moreover enjoys additional structural properties like invertibility, symmetry, or positive-definiteness, faster iterative methods including various Krylov subspace iterations and Conjugate-Gradient are available \cite{golub1996matrix}. 

Linear equations with binary coefficients are less ubiquitous than those with real coefficients,
but still they have key mathematical applications and commercial uses.
One of these is integer factorization,
in which we want to find large square numbers that are products of combinations of numbers whose prime factors are known, which means we want their exponents to be even \citep{pomerance1996tale,aboumrad2023lattice}. 
The solution of some Boolean satisfiability problems can be 
optimized using Gaussian elimination with binary arithmetic,
especially for cases with XOR constraints \citep{soos2009extending}. (Quantum approaches to more general Boolean satisfiability problems
are surveyed by \citet{alonso2022engineering}.)
Applications in cryptography and cryptanalysis, such as finding inverses of elements in the finite field $GF(2^m)$, have encouraged
research on optimizing solutions of binary linear systems 
\citep{rupp2006parallel,wang2016solving}. A component that solves linear systems modulo 2 is therefore 
required in various applications; choosing an optimal approach typically depends on features including
the dimension, number of equations, and sparsity; and research 
on this problem is ongoing \cite{hu2022universal}.

This paper presents a new quantum hybrid algorithm for solving 
this problem, which as far as we know is the first to apply quantum computing to solving binary linear systems.
The solution approach has two steps. First, we define a quantum circuit implementing matrix-vector products over the relevant finite vector spaces. Our circuit has one gate for each non-zero entry in the coefficient matrix. Then we derive a variational cost function that can be optimized in order to produce solutions to the given system. This frames the problem as one of optimizing a variational quantum circuit, which has become a standard approach in quantum machine learning
on NISQ (noisy intermediate-scale quantum) hardware \citep[Ch 5]{schuld2021machine}).

We deal with arbitrary unstructured matrices with entries in the finite field with two elements. Since the algorithm we present here leverages sparsity in the coefficient matrix, we comment on its computational complexity in relation to the state-of-the art block Wiedemann method in Section \ref{sec:related}. We note that the block Wiedemann method has been used in recent record-breaking RSA factorization calculations \cite{boudot2020difficulty}.

The rest of the paper is organized as follows. Section \ref{sec:intro_gates} introduces quantum gates and circuits used throughout the paper. Section \ref{sec:mat-vec-product} explains the construction of a quantum circuit that implements 
the crucial matrix-vector multiplication operation for binary-valued matrices and vectors. Section \ref{sec:vqa} pairs this circuit with a variational component that can be optimized in a hybrid quantum-classical computing framework in order to solve linear systems.
In particular, this section introduces a simple rotations ansatz which is especially
well-suited to the problem.
Section \ref{sec:results} presents initial proof-of-concept experiments, showing
that in low dimensions, with the rotations ansatz, the number of iterations needed to 
find a solution scales linearly in the number of dimensions (with a scaling factor around 4).
Section \ref{sec:related} compares the method introduced here with established classical and quantum alternatives.

We shall use the following notation throughout. We let $\F$ denote the field with the two elements ${0, 1}$. (In other works, $\F$ is sometimes
written as $\mathbb{Z}_2\cong\mathbb{Z}/2\mathbb{Z}$, or $GF(2)$, where
$GF$ stands for Galois Field.)
Addition in $\F$ is written using the symbol $\oplus$, so that $0 \oplus 0 = 0$, $0 \oplus 1 = 1 \oplus 0 = 1$, 
and $1 \oplus 1 = 0$. 
We fix an $m \times n$ matrix $A$ and an $m$-vector $b$ with $a_{ij}, b_i \in \F$. In addition, we let $x$ denote an arbitrary element of $\F^n$.
The problem of solving the linear system is to find $x$ such that $Ax = b$, in this case over $\F^m$ (so all 
arithmetic is modulo $2$). 

\section{Quantum Circuits and Gates Used in this Paper}
\label{sec:intro_gates}

This background section briefly reviews the quantum gates used in the 
circuits below. These include the single-qubit Pauli-$X$ gate and fractional $Y$-rotation gate $R_Y(\theta)$ (Figure \ref{fig:single-qubit-gates}), and the 2-qubit CNOT and controlled-$Z$ gates (Figure \ref{fig:cnot-gate}). 

The Pauli-$X$ gate is commonly used to flip a qubit between the $\ket{0}$ and $\ket{1}$ gates, which is why it is also sometimes called the quantum NOT gate. $X$-gates directed at different qubits can be used to prepare
an input state representing a binary-valued vector: the state 
$\ket{010 ... \ldots ... 001}$ is prepared by applying an $X$-gate to each of the qubits to be switched to the $\ket{1}$ state.

The Hadamard (H) gate is commonly used to put a qubit into a superposition state: for example, it maps a qubit
prepared in the state $\ket{0}$ to the superposition $\frac{1}{\sqrt{2}}(\ket{0} + \ket{1})$. 
Applying an H-gate to each qubit in an array is used to initialize a binary vector all of whose coordinates have
a 50-50 chance of being observed in the $\ket{0}$ or $\ket{1}$ state.

The CNOT gate is a 2-qubit entangling gate, that acts upon the
    state $\alpha\ket{00} + \beta\ket{01} + \gamma\ket{10} + \delta{\ket{11}}$. In the standard basis, its behavior can be described
    as ``performing a NOT operation on the target qubit if the control qubit is in state $\ket{1}$'', which in terms of operations in 
    $\F$ can be written as $\cnot \ket{y, z} = \ket{y, y \oplus z}$.

The Pauli-$X$, Hadamard, CNOT, and CX gates are self-inverse: performing these operations twice gives the identity map. These periodic properties are crucial for performing the binary arithmetic operations in the matrix-vector product of Section \ref{sec:mat-vec-product}. 

The $R_Y(\theta)$ rotates a single-qubit through an angle $\theta$ 
around the $Y$-axis on the Bloch sphere \citet[Ch 1]{nielsen2002quantumcomputation}. The angle $\theta$ can be varied,
and optimizing these angles for many gates is the task of the variational
algorithm in Section \ref{sec:vqa}. The controlled-$Z$ gate is similar 
to the CNOT gate and entangles 2 qubits: in the standard basis, its
action is symmetric (in the sense that it doesn't matter which qubit is considered to be the control and which the target qubit).

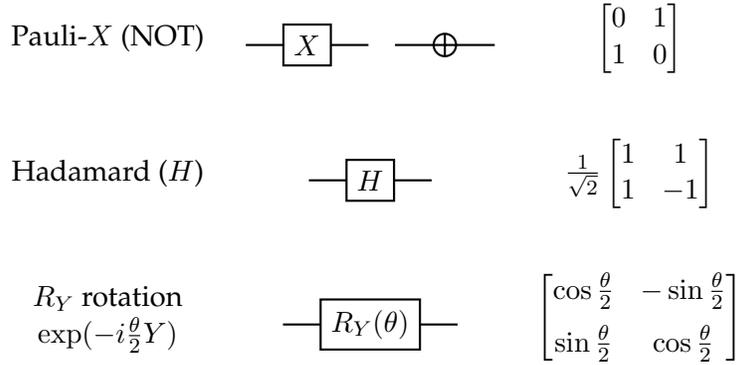
\begin{figure}
    \centering

\begin{tabular}{ccc}
Pauli-$X$ (NOT)
& 
\begin{tikzcd} \qw & \gate{X} & \qw \end{tikzcd}
\begin{tikzcd} \qw & \targ{} & \qw \end{tikzcd}
&
$\begin{bmatrix} 0 & 1 \\ 1 & 0 \end{bmatrix}$
\\[6 ex]

Hadamard ($H$)
& 
\begin{tikzcd} \qw & \gate{H} & \qw \end{tikzcd}
&
$\frac{1}{\sqrt{2}}\begin{bmatrix} 1 & 1 \\ 1 & -1 \end{bmatrix}$
\\[6 ex]

\begin{tabular}{c}
$R_Y$ rotation \\
$\exp(-i\frac{\theta}{2}Y)$

\end{tabular}
& 
\begin{tikzcd} \qw & \gate{R_Y(\theta)} & \qw \end{tikzcd}
&
$\begin{bmatrix} 
\cos \frac{\theta}{2} & -\sin \frac{\theta}{2} \\[6pt] 
\sin \frac{\theta}{2} & \cos \frac{\theta}{2}
\end{bmatrix}$

\\[6 ex]

\end{tabular}    
    \caption{Single-qubit gates used in this paper, and their corresponding matrices,  
    which operate on the superposition state $\alpha\ket{0} + \beta\ket{1}$ written as the column vector $\begin{bmatrix} \alpha & \beta \end{bmatrix}^T$.}
    \label{fig:single-qubit-gates}
\end{figure}

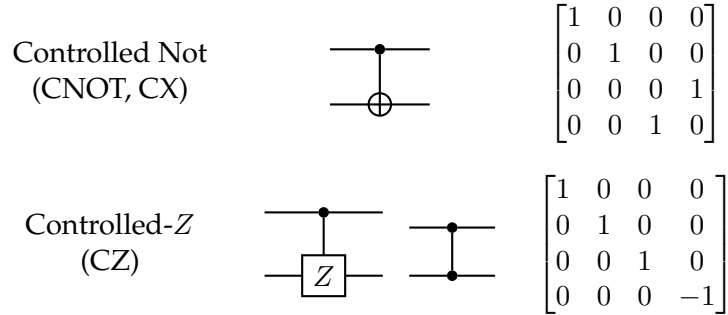
\begin{figure}
    \centering
\begin{tabular}{ccc}
\begin{tabular}{c}
Controlled Not \\ (CNOT, CX)
\end{tabular}
& 
\begin{tikzcd} 
\qw & \ctrl{1} & \qw \\
\qw & \targ{} & \qw
\end{tikzcd}
&
$\begin{bmatrix}
1 & 0 & 0 & 0 \\
0 & 1 & 0 & 0 \\
0 & 0 & 0 & 1 \\
0 & 0 & 1 & 0 
\end{bmatrix}$
\\[6 ex]

\begin{tabular}{c}
Controlled-$Z$ \\ (CZ)
\end{tabular}
& 
\begin{tikzcd} 
\qw & \ctrl{1} & \qw \\
\qw & \gate{Z} & \qw
\end{tikzcd}
\begin{tikzcd} 
\qw & \ctrl{1} & \qw \\
\qw & \control{} & \qw
\end{tikzcd}
&
$\begin{bmatrix}
1 & 0 & 0 & 0 \\
0 & 1 & 0 & 0 \\
0 & 0 & 1 & 0 \\
0 & 0 & 0 & -1 
\end{bmatrix}$
\\

\end{tabular}    
    \caption{Two-qubit CNOT (controlled-$X$) and controlled-$Z$ gates}
    \label{fig:cnot-gate}
\end{figure}

\section{Implementing the Matrix-Vector Product as a Quantum Circuit}
\label{sec:mat-vec-product}

This section introduces a quantum circuit that implements the binary-valued matrix-vector product $Ax$, where $A$ is a matrix and $x$ is a vector, and both have values in $\F$ as defined above.

The trick is to notice that binary arithmetic may be implemented with controlled-NOT operations. The NOT gate switches the state of an individual qubit between $\ket{0}$ and $\ket{1}$,
and the $\cnot$ gate performs such an operation only if the so-called \textit{control} qubit is in state $\ket{1}$. In particular, the $\cnot$ gate acts on the two-qubit state $\ket{y, z}$ as
\begin{equation}\label{eq:cnot action}
\cnot \ket{y, z} = \ket{y, y \oplus z}.
\end{equation}

Now let $\ket{x} = \ket{x_1, x_2, \ldots, x_n}$ denote the $n$-qubit quantum state corresponding to $x$ and notice that
\[
Ax \equiv
\begin{bmatrix}
    a_{11} x_1 \oplus \cdots \oplus a_{1n} x_n \\
    \vdots \\
    a_{m1} x_1 \oplus \cdots \oplus a_{mn} x_n
\end{bmatrix}
\mod 2.
\]
Thus the $m$-qubit state $\ket{Ax}$ can be prepared by applying the quantum circuit
\begin{equation}\label{eq:mat-vec prod}
    \mathcal{A} = \prod_{j = 1}^n \mathcal{A}_j,
    \quad \text{with} \quad
    \mathcal{A}_j = \prod_{i=1}^m \cnot(j, n + i)^{a_{ij}}
\end{equation}
to the tensor product $\ket{x}\ket{0}$. 
Here $\cnot(k, \ell)$ denotes a $\cnot$ gate controlled by the $k^{\mathrm{th}}$ qubit and targeting the $\ell^{\mathrm{th}}$ one.
The product operator $\Pi$ in these definitions denotes composition, so it is implemented by applying the gates in sequence.
Both products in Relation \eqref{eq:mat-vec prod} may be taken in any order. We note that $\mathcal{A}$ requires $m + n$ qubits and $N$ quantum gates, with $N$ denoting the number of non-zero entries in $A$.

In particular, we obtain the following theorem.

\begin{thm}\label{thm:mat-vec qc}
    Using $\mathcal{A}$ as in Relation \eqref{eq:mat-vec prod}, with $\ket{0}$ denoting the $m$-qubit all-zero state, and
    with $\ket{x}\ket{0}$ denoting the concatenation / tensor product of $\ket{x}$ and $\ket{0}$,
    we have 
\[
    \mathcal{A} (\ket{x}\ket{0}) = \ket{x} \ket{Ax}.
\]
\end{thm}
\begin{proof}
    The proof follows from an easy argument by induction on $n$ that is left to the reader: the base case $n = 1$ can be easily established for all $m$ by verifying that the $i^{\mathrm{th}}$ qubit in the output register is zero precisely when $a_{ij} = 0$, and the inductive step follows from a quick calculation considering the inductive hypothesis and the action of the $\cnot$ gate, as in Relation \eqref{eq:cnot action}.
\end{proof}

\begin{eg}
\label{ex:linear_system}
The following example elucidates Theorem \ref{thm:mat-vec qc}. For concreteness, consider the $2 \times 3$ matrix
\[
    A = \begin{bmatrix}
        1 & 0 & 1 \\
        1 & 1 & 0
    \end{bmatrix}
    \quad \text{so that} \quad
    A x \equiv \begin{bmatrix}
        x_1 \oplus x_3 \\
        x_1 \oplus x_2
    \end{bmatrix} 
    \mod 2.
\]
Then
\begin{align*}
    \mathcal{A}_1 = \cnot(1, 4) \cnot(1, 5), \quad
    \mathcal{A}_2 = \cnot(2, 5), \quad \text{and} \quad
    \mathcal{A}_3 = \cnot(3, 4),
\end{align*}
and therefore
\begin{align*}
    \mathcal{A} \ket{x} \ket{0, 0} = \mathcal{A}_3 \mathcal{A}_2 \ket{x} \ket{x_1, x_1} 
    = \mathcal{A}_3 \ket{x} \ket{x_1, x_1 \oplus x_2} 
    = \ket{x} \ket{x_1 \oplus x_3, x_1 \oplus x_2}
    = \ket{x, Ax}.
\end{align*}

The diagram in Figure \ref{fig:mat_vec_prod_qc} shows how the operator $\mathcal{A}$ is implemented explicitly as a quantum circuit.

\begin{figure}[H]
    \centering
\includegraphics[width=1.8in]{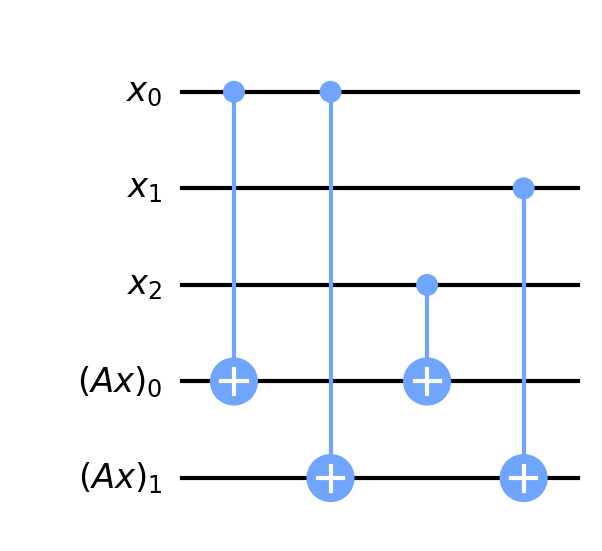}
    \caption{Quantum circuit implementing the operator $\mathcal{A}$}
    \label{fig:mat_vec_prod_qc}
\end{figure}
\centering
\end{eg}

These circuits work because the self-inverse behavior of the $X$ and CNOT gates is ideal for binary arithmetic over $\F$. 
It is possible that the periodic nature of other gates could be used similarly to perform arithmetic operations over other finite fields, but such an adaptation is not straightforward unless we consider qu\textit{d}its instead of qu\textit{b}its, because the way several CNOT gates
combine angles into a single target qubit does not perform like group addition for fields other than $\F$.
More precisely, the $X$ and $\cnot$ gates rotate their target qubits through angles $0$ and $\pi$, which is all
we need for the $0$ and $1$ elements of $\F$, but if this is extended to a larger set of angles $2\pi / p$ for $p\neq 2$, the use of CNOT gates to combine different contributions into a single target qubit as in Figure \ref{fig:adder_circuit} is nonlinear \cite{widdows2022near}. It follows that the coordinates of the output vector $b$ are not the same as the sums
of the various inputs $(Ax)_j$, except for the 2-element field $\F$.

\begin{figure}[H]
    \centering
    \includegraphics[width=6cm]{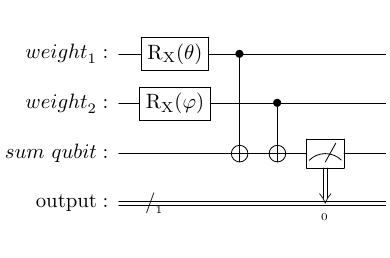}
    \caption{An `adder circuit' that combines the values of $\theta$ and $\phi$ into the target qubit, whose
    probability of measuring a $\ket{1}$ is $\sin^2(\theta)\cos^2(\varphi) + \cos^2(\theta)\sin^2(\varphi)$, as demonstrated in \cite{widdows2022nonlinear}. Note the structural similarity with the inputs to $(Ax)_0$ in Figure
    \ref{fig:mat_vec_prod_qc}.}
    \label{fig:adder_circuit}
\end{figure}

\section{Solving the Linear System using a Variational Quantum Algorithm}
\label{sec:vqa}

Now we propose a Variational Quantum Algorithm (VQA) designed to solve the linear system $A x \equiv b \mod 2$. 
As with any VQA, there are two main ingredients: a variational ansatz, and a cost function that serves as the optimization objective. 
The ansatz gives a circuit pattern or template, with gate parameters (typically angles) that can be varied and optimized.
In this report we consider one objective function and compare two different ansatze: a Rotations ansatz specially designed for the task,
and a more generic Brickwork ansatz.

The system as a whole is called $\mathrm{Mod2VQLS}$, which stands for Modulo-2 Variational Quantum Linear Solver.
Figure \ref{fig:mod2_vqls_circuit} illustrates the full variational circuit evaluated by $\mathrm{Mod2VQLS}$ when solving the linear system described in Example \ref{ex:linear_system}. In this case, the brickwork layout variational ansatz described in Section \ref{subsec:brickwork} is used, with $4$ layers.

\begin{figure}[!t]
    \centering
    \includegraphics[width=0.9\textwidth]{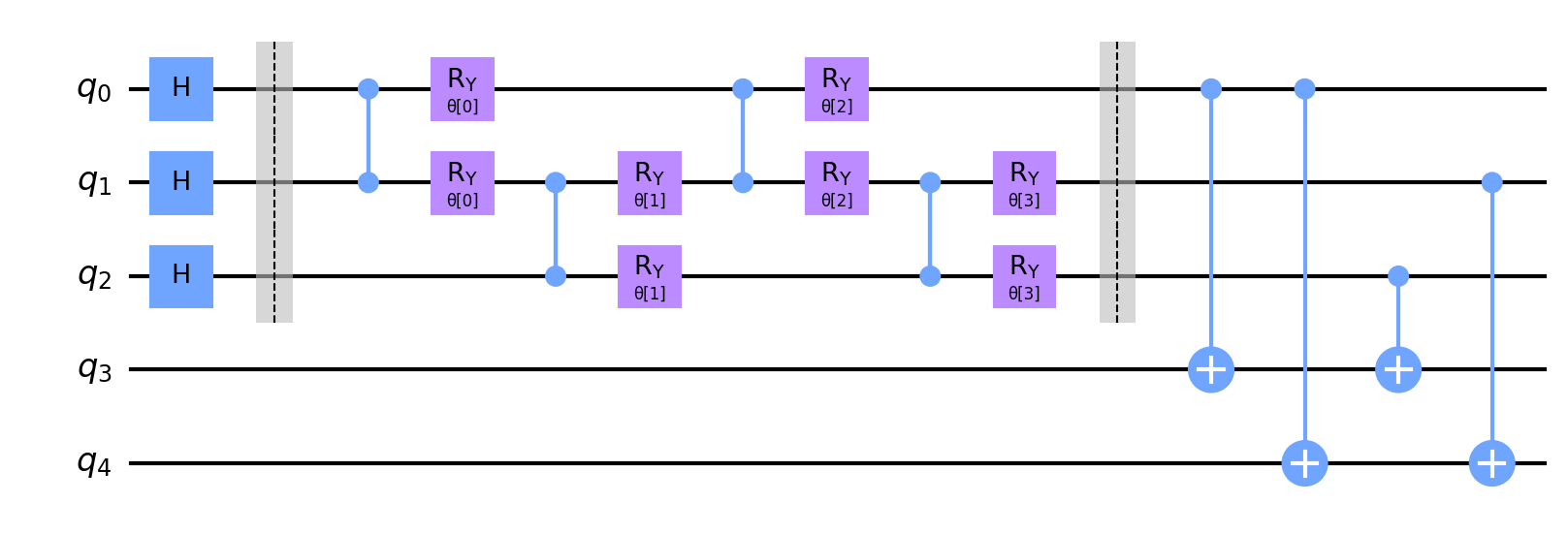}
    \caption{Variational circuit for solving $Ax = b$ as in Example \ref{ex:linear_system} using a brickwork layout ansatz with $2$ layers.}
    \label{fig:mod2_vqls_circuit}
\end{figure}

Our cost function here measures the overlap between the projector $\ket{\psi(\theta)}\bra{\psi(\theta)}$ and the subspace orthogonal to $\ket{b}$, which is given by
\begin{equation*}
    \widehat{C} = \mathrm{tr}\big(\ket{\psi(\theta)}\bra{\psi(\theta)}(\mathrm{id} - \ket{b}\bra{b})\big).
\end{equation*}
We note that this cost function has appeared before in \citep{bravoprieto2020variational} as Equation (3) in the setting of solving square linear systems with \textit{real} entries.

Some algebra shows that $\widehat{C}$ can be evaluated by computing the expected energy of an Ising Hamiltonian:
\begin{align}
\label{eq:cost as overlap}
    \widehat{C} &= \mathrm{tr}\big(\ket{\psi}\bra{\psi}(\mathrm{id} - \ket{b}\bra{b})\big) 
    = \bra{\psi}(\mathrm{id} - \ket{b}\bra{b})\ket{\psi}
    = 1 - \vert \braket{b \,\vert\, \psi} \vert^2,
\end{align}
where we have omitted the dependence on $\theta$ for simplicity of notation. Notice that the second term is the expected value of the rank-$1$ Ising Hamiltonian $\ket{b}\bra{b}$ with eigenvalue $1$ corresponding to the eigenstate $\ket{b}$, computed with respect to the variational state $\ket{\psi}$. Although it may be difficult to express this Hamiltonian as a linear combination of tensor products of Pauli matrices, we do not need to construct or even to know it explicitly: all we need is an oracle that can lazily evaluate the Hamiltonian's eigenvalues.

\subsection{Rotations Ansatz}
\label{subsec:rotations}

Now we turn to our variational ansatze. The first is a simple rotations pattern
that is particularly well-suited for the Mod2VQLS cost function $\widehat{C}$.
This ansatz is just a product of single-qubit rotations about the $Y$-axis; in particular, we take
\begin{equation}
\label{eq:rotation ansatz}
    V(\theta) = \prod_{j = 1}^n R_Y^j(\theta_j),
\end{equation}
with $R_Y^j$ denoting an $R_Y$ rotation applied on the $j^{\mathrm th}$r qubit and each parameter 
$\theta \in [-2\pi, 2\pi]^n$. This circuit is interesting because it does not add to the overall computational cost of each iteration, and it is amenable to direct mathematical analysis. 

In particular, we can derive an explicit formula for the variational cost as a function of the circuit parameters.

\begin{thm}
    For each $x \in \mathbb{F}_2^n$, define $\alpha_x(\theta) = \prod_{j=1}^n \cos(\theta_j/2)^{(1 - x_j)} \sin(\theta_j/2)^{x_j}$. In addition, let $A^{-1}(b) = \{x \in \mathbb{F}_2^n \mid Ax \equiv b \mod 2\}$ denote the inverse image of $b$ under $A$. Then
    $$\widehat{C}(\theta) = 1 - \sum_{x \in A^{-1}(b)} \alpha_x(\theta)^2$$
    with $\ket{\psi(\theta)} = \mathcal{A} V(\theta) \ket{0}$ and $V(\theta)$ denoting the ansatz defined by Relation \eqref{eq:rotation ansatz}.
\end{thm}
\begin{proof}
    The theorem follows from direct calculation. First note that \begin{equation*}
        R_Y(\theta_j)\ket{0} = \cos(\theta_j/2) \ket{0} + \sin(\theta_j/2) \ket{1}.
    \end{equation*}
    So a quick induction, left to the reader as an exercise, shows that
    \begin{equation*}
        V(\theta) \ket{0}^{\otimes n} 
        = \bigotimes_{j=1}^n \big(R_Y^j(\theta_j) \ket{0}\big) 
        = \bigotimes_{j=1}^n \big(\cos(\theta_j/2) \ket{0} + \sin(\theta_j/2) \ket{1}\big)
        = \sum_{x \in \mathbb{F}_2^n} \alpha_x(\theta) \ket{x}.
    \end{equation*}

    Combined with Theorem \ref{thm:mat-vec qc}, the last equation yields an expression for our variational quantum state:
    \begin{equation}
    \label{eq:variational state}
        \ket{\psi(\theta)} = \mathcal{A} V(\theta) \ket{0}^{\otimes n} = \sum_{x \in \mathbb{F}_2^n} \alpha_x(\theta) \ket{x} \ket{Ax}.
    \end{equation}

    We use this expression to compute the variational cost with respect to $\ket{\psi(\theta)}$. In particular, let 
    \begin{equation*}
        \mathcal{H} = I \otimes \ket{b} \bra{b} 
        = \sum_{y \in \mathbb{F}_2^n} \ket{y}\bra{y} \otimes \ket{b} \bra{b}
    \end{equation*}
    and use Relation \eqref{eq:cost as overlap} to write
    \begin{equation*}
        1 - \widehat{C}(\theta) = |\langle b \,\vert\, \psi(\theta)\rangle|^2
        = \bra{\psi(\theta)} \mathcal{H} \ket{\psi(\theta)}.
\end{equation*}
    Now since $\alpha_x(\theta)$ is real, we see that
    \begin{align*}
        \bra{\psi(\theta)} \mathcal{H} \ket{\psi(\theta)}
        &= \sum_{x, y} \alpha_x(\theta) \alpha_y(\theta) \langle x \,\vert\, y \rangle \langle Ax \,\vert\, b \rangle \langle b \,\vert\, Ay \rangle \\
        &= \sum_{x \in \mathbb{F}_2^n} \alpha_x(\theta)^2 |\langle Ax \,\vert\, b \rangle|^2 \\
        &= \sum_{x \in A^{-1}(b)} \alpha_x(\theta)^2,
    \end{align*}
    so $\widehat{C}(\theta) = 1 - \sum_{x \in A^{-1}(b)} \alpha_x(\theta)^2$ as desired.
\end{proof}

Moreover, we can obtain a similar formula for the gradient of the cost function with respect to the circuit parameters.
\begin{thm}
\label{thm:cost grad}
For each $j = 1, \ldots, n$, the cost function $\widehat{C}(\theta)$ varies like
\begin{equation*}
    \frac{\partial}{\partial \theta_j}\widehat{C}(\theta) = \sum_{x \in A^{-1}(b)} (-1)^{1 - x_j} \alpha_x(\theta) \alpha_{x \oplus e_j}(\theta)
\end{equation*}
with respect to $\theta_j$. Here $\oplus$ denotes binary addition over $\F^n$ and $e_j$ denotes the $j^{\mathrm{th}}$ standard basis vector in $\F^n$.
\end{thm}
\begin{proof}
    The theorem follows from a direct computation using the chain rule. The key is to notice that differentiating $\alpha_x$ swaps a sine for a cosine and vice versa; for instance, if $x_j = 0$, and we fix $c_k = \cos\big(\frac{\theta_k}{2}\big)$ and $s_k = \sin\big(\frac{\theta_k}{2}\big)$ for simplicity of notation, we see that
    \begin{equation*}
        \frac{\partial}{\partial \theta_j}\alpha_x(\theta) 
        = \frac{\partial}{\partial \theta_j}\bigg[ \prod_{k=1}^n c_k^{1 - x_k} s_k^{x_k}\bigg] 
        = \frac{-s_j}{2} \bigg[ \prod_{k \neq j} c_k^{1 - x_k} s_k^{x_k}\bigg] 
        = -\frac{1}{2} \alpha_{x \oplus e_j}(\theta).
        \qquad\qquad\qedhere
    \end{equation*}
\end{proof}

Theorem \ref{thm:cost grad} helps us understand the critical points on the variational cost landscape. In particular, it shows that the gradient $\nabla \widehat{C}(\theta)$ is smooth, and moreover, each of its entries is a trigonometric polynomial on the $2n$ variables $c_j = \cos\big(\frac{\theta_j}{2}\big)$ and $s_j = \sin\big(\frac{\theta_j}{2}\big)$, for $j = 1, \ldots, n$. Thus every critical point of the cost surface satisfies $\nabla \widehat{C} = 0$. When combined with the Pythagorean identities $c_j^2 + s_j^2 = 1$, the equation $\nabla \widehat{C} = 0$ characterizes an algebraic variety defined by $2n$ polynomials in $2n$ variables. Hence Bezout's Theorem 
\citep[Ch 2]{fischer2001plane}
implies $\widehat{C}$ has at most $(2n)^n \cdot 2^n = (4n)^n$ real extrema.

It is interesting to note that when $(c_j, s_j)$ is an $8n$-th root of unity, our variational cost can be written in terms of \textit{quantum integers} or \textit{$q$-integers}, which are ubiquitous in $q$-calculus, the representation theory of quantized enveloping algebras and quantum groups, and in the theory of crystal bases, amongst others \cite{kac2015qcalculus,kassel1995qgroups,charipressley2000qgroups,kashiwara1990crystals,lusztig1990canonbases}. We conjecture that these points characterize the extrema of our cost function. The conjecture is motivated by the size of the variety characterized by $\nabla \widehat{C} = 0$ together with the Pythagorean identities as computed by Bezout's Theorem and by the explicit computations summarized by Theorem \ref{thm:optimal params}.

\begin{defn}
    For any complex $q \neq 1$ and any integer $n$, the \textit{quantum integer} or \textit{$q$-integer} $[n]_q$ is defined by
    \begin{equation*}
        [n]_q = \frac{q^n - q^{-n}}{q - q^{-1}} = q^{n-1} + q^{n - 3} + \cdots + q^{3 - n} + q^{1-n}.
    \end{equation*}
\end{defn}

In this setting, quantum integers appear in the coefficients $\alpha_x(\theta)$. In particular, let $\theta_j = \frac{\pi}{2n}p_j$ for any $p_j \in \{0, 1, \ldots, 4n-1\}$, let $\xi = e^{\frac{i \pi}{4n}}$ denote a primitive $8n$-th root of unity, and notice
\begin{equation*}
    \cos(\theta_j/2)^{1 - x_j} \sin(\theta_j/2)^{x_j}
    = \frac{\big(i^{1 - x_j} e^{\frac{i \pi \theta_j}{2}}\big) - \big(i^{1 - x_j} e^{\frac{i \pi \theta_j}{2}}\big)^{-1}}{2i}
    = \frac{[2n(1 - x_j) + p_j]_\xi}{[2n]_\xi}.
\end{equation*}

We use the last identity to prove the existence of globally optimal parameters for our variational circuit. It is important to observe these guarantees are not typically offered by VQAs; for instance, QAOA can only guarantee such parameters \textit{in the limit of infinite circuit depth}.

We will need the following identities.
\begin{lem}
\label{lem:q int rels}
    If $\xi = e^{\frac{i\pi}{4n}}$ denotes a primitive $8n$-th root of unity, then the following identities relating $\xi$-integers hold:
    \begin{equation*}
        [3n]_\xi = [n]_\xi = \frac{[2n]_\xi}{\sqrt{2}}, 
        \quad\text{and}\quad
        [4kn]_\xi = 0 \quad\text{for any}\quad k \in \mathbb{Z}.
    \end{equation*}
\end{lem}
\begin{proof}
    The lemma follows from direct calculation. For example, 
    \begin{equation*}
        [3n]_\xi = \frac{\xi^{3n} - \xi^{-3n}}{\xi - \xi^{-1}} 
        = \frac{2i}{\xi - \xi^{-1}} \cdot \sin\big(3n \cdot \frac{\pi}{4n}\big) 
        = \frac{2i}{\xi - \xi^{-1}} \cdot \sin\big(n \cdot \frac{\pi}{4n}\big)
        = [n]_\xi.
    \end{equation*}
    The second equality follows because $\sin\big(n\cdot \frac{\pi}{4n}\big) = 1/\sqrt{2}$ and $\xi^{2n} = i$, so that $\xi^{2n} - \xi^{-2n} = 2i$. The last equality holds because $\xi$ is an $8n$-th root of unity; in particular, $\xi^{4kn} = \xi^{-4kn}$.
\end{proof}

\begin{thm}
\label{thm:optimal params}
    Let $\theta = \frac{\pi}{2n} p$ for some $p \in \mathbb{Z}^n$. First suppose $p = 2n x'$ for some $x' \in \{0, 1\}^n$. Then $\widehat{C}(\theta) = 0$ is a global minimum of $\widehat{C}\colon \mathbb{R}^n \to \mathbb{R}$ if $x'$ satisfies $Ax' = b$; otherwise, $\widehat{C}(\theta) = 1$ is a global maximum of the cost function.
    
    Now assume $p_1 = \cdots = p_n = k n$ for some odd integer $k$ and let $\mathrm{rk}(A)$ denotes the rank of $A$ over $\F$. In this case $\widehat{C}(\theta) = 1 - 2^{-\mathrm{rk}(A)}$.
\end{thm}
\begin{proof}
    First notice our cost function is non-negative; for any $\theta \in \mathbb{R}^n$, we have
    \begin{equation*}
        \widehat{C}(\theta) = 1 - \sum_{x \in A^{-1}(b)} \alpha_x(\theta)^2 
        \geq 1 - \sum_{x \in \mathbb{F}_2^n} \alpha_x(\theta)^2
        = 1 - \langle \psi(\theta) \,\vert\, \psi(\theta) \rangle = 0.
    \end{equation*}
    This follows from Relation \eqref{eq:variational state}, which writes $\ket{\psi(\theta)}$ as a superposition over a tensor product of computational basis states with amplitudes given by $\alpha_x(\theta)$.
    In addition, $\widehat{C}$ is bounded above by $1$ because each $\alpha_x(\theta)^2$ is non-negative. This means $0 \leq \widehat{C}(\theta) \leq 1$ for every $\theta \in \mathbb{R}^n$.
    
    Thus, to prove the first claim, it remains to show that when $p = 2n x'$, the cost function vanishes if $x'$ belongs to the inverse image $A^{-1}(b)$ and it is unity otherwise. This follows from direct calculation:
    \begin{align*}
        1 - \widehat{C}(\theta) &= \sum_{x \in A^{-1}(b)} \alpha_x(\theta)^2
        = \frac{1}{[2n]_\xi^{2n}} \sum_{x \in A^{-1}(b)} \prod_{j = 1}^n [2n(1 - x_j + x_j^*)]_\xi^2 \\
        &= \frac{1}{[2n]_\xi^{2n}} \sum_{x \in A^{-1}(b)} \delta_{x, x'} [2n]_\xi^{2n} \\
        &= \sum_{x \in A^{-1}(b)} \delta_{x, x'}.
    \end{align*}
    Note that the product in the second equality vanishes whenever $x \neq x'$ because $[4n]_\xi = 0$.

    For the second claim, write $k = 4m + r$ for some integer $m$ and with $r = 1, 3$. Now recall the well-known ``summation rule'' $[s + t]_q = q^s [t] + q^{-t} [s]_q$ for quantum integers \cite[Relation~V1.1.2]{kassel1995qgroups}. When combined with Lemma \ref{lem:q int rels}, the summation rule implies
    \begin{equation}
    \label{eq:q int sum}
        [kn]_\xi = \xi^{rn} [4mn]_\xi + \xi^{-4mn} [rn]_\xi = (-1)^m [rn]_\xi.
    \end{equation}
    In addition, note that if $r = 3$, Relation \ref{eq:q int sum} implies $[(r+2)n]_\xi = [5n]_\xi = -[rn]_\xi$. Thus if $\mathbf{1}^T x \in \mathbb{Z}$ denotes the (integer) sum of the elements of $x$, we see that
    \begin{align*}
        1 - \widehat{C}(\theta) &= \sum_{x \in A^{-1}(b)} \alpha_x(\theta)^2
        = \frac{1}{[2n]_\xi^{2n}} \sum_{x \in A^{-1}(b)} \prod_{j = 1}^n [2n(1 - x_j) + kn]_\xi^2 \\
        &= \frac{1}{[2n]_\xi^{2n}} \sum_{x \in A^{-1}(b)} \big([(k + 2)n]_\xi^{n - \mathbf{1}^T x} [kn]_\xi^{\mathbf{1}^T x}\big)^2 \\
        &= \frac{1}{[2n]_\xi^{2n}} \sum_{x \in A^{-1}(b)} \big([(r + 2)n]_\xi^{n - \mathbf{1}^T x} [rn]_\xi^{\mathbf{1}^T x}\big)^2 \\
        &= \frac{1}{[2n]_\xi^{2n}} \sum_{x \in A^{-1}(b)} [n]_\xi^{2n} \\
        &= \frac{|A^{-1}(b)|}{\sqrt{2}^{2n}}.
    \end{align*}
    In the fourth equality we used Lemma \ref{lem:q int rels} to simplify $[(r + 2)]_\xi^2 = [rn]_\xi^2 = [n]_\xi^2$. In the fifth equality we applied Lemma \ref{lem:q int rels} again, this time substituting $[n]_\xi = [2n]_\xi / \sqrt{2}$. 

    To conclude, we recall the solution set $A^{-1}(b)$ is an affine space in $\F^n$ with one point for each element of the kernel of $A$. Hence if we let $\mathrm{null}(A)$ denote the dimension of $\ker{A}$ over $\F$ and we recall the Rank-Nullity Theorem we obtain the desired result:
    \begin{equation*}
        \widehat{C}(\theta) = 1 - \frac{|A^{-1}(b)|}{\sqrt{2}^{2n}} = 1 - \frac{2^{\mathrm{null}(A)}}{2^n} = 1 - 2^{-\mathrm{rk}(A)}. 
        \qquad\qquad\qedhere
    \end{equation*}
\end{proof}

\subsection{Brickwork Ansatz}
\label{subsec:brickwork}

The other variational ansatz we consider is comprised by a brickwork layout of parametrized two-qubit gates. For instance, Figure \ref{fig:brickwork_ansatz} illustrates a brickwork layout ansatz with a depth of $5$ layers on $4$ qubits. The two-qubit ``brick'' 
design was used by authors including \cite{bravoprieto2020variational} for general linear solving, and \cite{niu2022holographic}
for electron simulation. Various other general-purpose ansatz designs could be tried.

\begin{figure}[!ht]
    \centering
    \includegraphics[width=0.9\textwidth]{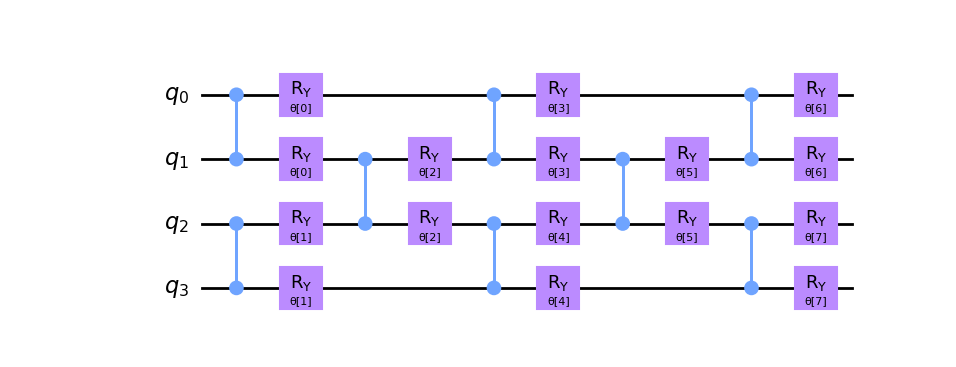}
    \caption{A brickwork layout ansatz on $4$ qubits with a depth of $5$ layers.}
    \label{fig:brickwork_ansatz}
\end{figure}

With enough layers, it should be possible for the brickwork ansatz to provide solutions to the equation $Ax = b$,
and intuitively, we expected that convergence would be slower but that the method might find more solutions overall.
In practice, the convergence was considerably slower, and in most cases, the number of extra solutions found was modest.

\subsection{Alternatives}

There are several alternatives and ways this design could be varied.
For instance, one could use a different two-qubit block in the brickwork layout ansatz. 
More generally, we could use a different ansatz altogether. In particular, it might be beneficial to incorporate linear-algebraic information into the ansatz, such as something that encourages the superposition to be orthogonal to the range of $A$, in case $b = 0$.
Also, a swap test could be used for computing $\widehat{C}$ as part of the quantum circuit itself \cite{buhrman2001quantum}. 
This may serve as a way to mitigate the measurement error, at the cost of requiring additional qubits.

\section{Experimental Results}
\label{sec:results}

This section describes the results of experiments evaluating the performance of Mod2VQLS solvers using both the rotation and brickwork ansatze.

In each dimension from 1 to 9, we generated ten consistent binary linear systems by constructing $n \times n$ matrices $A$ and $n$-vectors $x$ with independent and uniformly selected entries in $\{0, 1\}$ and then computing $Ax = b$. Then we solved the systems using Mod2VQLS, using both the brickwork and rotation ansatze. Solving a linear system entails optimizing the model's variational parameters and then measuring the optimized state $\ket{\psi(\theta^*)}$. We tested all the computational states observed in the optimized superposition in order to determine whether they were valid solutions to the linear system in question. We counted the number of distinct valid and invalid solutions proposed, as well as the average number of iterations required for convergence of the variational method. We used SciPy's COBYLA implemetation to update the variational parameters in our circuits; we note this optimization routine only requires a single quantum circuit execution per iteration. Circuits were coded and simulated using the Qiskit Python package \cite{qiskit2021textbook}.

In the Brickwork case, the number of layers in the variational ansatz is an extra configuration parameter. From experimenting, we found that at least 2 layers were needed, and that matching the number of layers / parameters to the number of dimensions produced a reasonable tradeoff between solution iterations and correctness.

Our results are presented in Table \ref{tab:small_dim_results}. The rotations ansatz performed quite simply and effectively, always finding a correct solution, and occasionally finding more than one, with a reasonably small number of iterations. The rotations ansatz proposed no invalid solutions.
The brickwork ansatz was more costly and error-prone, using more iterations and producing some invalid solutions. The brickwork ansatz was also able to find a greater variety of valid solutions, though not dramatically so.

\begin{table}[t]
    \centering
    \begin{tabular}{|c|rrrr|rrrr|}
    \hline
 & \multicolumn{4}{c|}{\parbox[r]{3cm}{\quad \\ Brickwork Ansatz\\ \quad}} 
 &  \multicolumn{4}{c|}{\parbox[r]{2.5cm}{\quad \\ Rotation Ansatz\\ \quad}} \\
 Dim  & Solved & Valid & Invalid & Average & Solved & Valid & Invalid & Average \\
 & & Solns & Solns & Iterations & & Solns & Solns & Iterations \\
\hline
1 & 8  & 8 & 2 & 1 & 10 & 10 & 0 & 2 \\
2 & 10 & 16 & 4 & 1.7 & 10 & 14 & 0 & 3.7 \\
3 & 10 & 14 & 5 & 1 & 10 & 11 & 0 & 9.2 \\
4 & 9  & 9 & 1 & 38.8 & 10 & 11 & 0 & 16.3 \\
5 & 10 & 19 & 0 & 93.2 & 10 & 11 & 0 & 17.3 \\
6 & 10 & 15 & 0 & 177.4 & 10 & 11 & 0 & 18.9 \\
7 & 10 & 11 & 1 & 86.2 & 10 & 10 & 0 & 24.6 \\
8 & 10 & 11 & 0 & 120.2 & 10 & 10 & 0 & 29.4 \\
9 & 10 & 10 & 0 & 122.4 & 10 & 10 & 0 & 33.5 \\
\hline
    \end{tabular}
    \caption{Experimental results using Brickwork and Rotation Ansatze in various dimensions,
    using 10 randomly-generated matrices in each dimension, showing how many valid and invalid solutions
    were proposed, and how many iterations used, by each method.} 
    \label{tab:small_dim_results}
\end{table}

Figure \ref{fig:iterations_per_dim} show the growth in the number of iterations used
by Mod2VQLS with the rotations ansatz. The growth is roughly linear in small dimensions,
with a slope of (just under) 4.
This allows for comparison with block Wiedemann, as explained in Section \ref{sec:related} below.

\begin{figure}
    \centering
    \includegraphics[width=8cm]{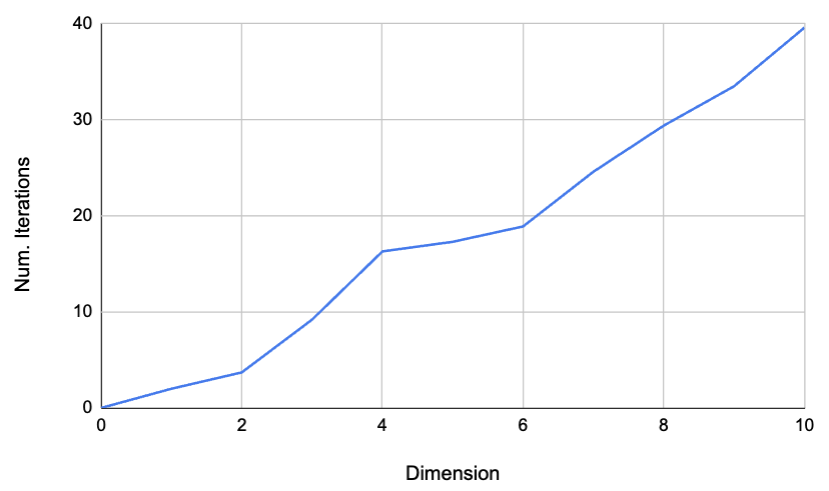}
    \caption{Average number of iterations needed to find a solution for each dimension using the rotations ansatz. The number of iterations
    grows approximately linearly, with a slope just under 4.}
    \label{fig:iterations_per_dim}
\end{figure}

These experiments demonstrate proof-of-concept, in that the Mod2VQLS system does find correct solutions.
They also demonstrate that there are performance tradeoffs, and that different choices of ansatz may 
be appropriate, depending on whether the task requires for searching for any solution, or
a more exhaustive search for all solutions.

\section{Related Classical and Quantum Methods for Solving Linear Systems}
\label{sec:related}

To begin we comment on the computational complexity of our $\mathrm{Mod2VQLS}$ to the block Wiedemann method when $m = n$, which is the fastest known classical algorithm for dealing with sparse unstructured systems over finite fields \cite{coppersmith1994blockwiedemann}. This method requires a linear number of matrix-vector multiplications plus a quadratic number of arithmetic operations in $\F$. In fact, large scale implementations,  like the one leveraged in the record-breaking RSA-$240$ calculation described in \cite{boudot2020difficulty}, have used $\geq 3n$ matrix-vector multiplications and it is not known if less than $2n$ can be used \cite{kaltofen1995wiedemannanalysis}. In our quantum setting, we consider the cost of executing our matrix-vector product circuit $\mathcal{A}$ from Section \ref{sec:mat-vec-product}. The results from our numerical experiments using the rotations ansatz described in Section \ref{subsec:rotations}, as recorded in Table \ref{tab:small_dim_results}, indicate that the gate complexity of our simple quantum algorithm is on par with the state-of-the-art block Wiedemann approach: the number of matrix-vector multiplications needed to achieve a solution is linear in the system dimension and we do not require the additional quadratic number of field operations or linear storage. (However, our method requires at least one qubit per variable, so large-scale systems encountered in practice will require larger quantum computers.)

Now we turn to highlighting some important distinctions between $\mathrm{Mod2VQLS}$ and quantum linear solvers that have appeared previously, like the VQLS algorithm in \citep{bravoprieto2020variational} and the renowned HHL algorithm \citep{harrow2009quantum} (some of whose variants and improvements are described in \citep{dervovic2018quantum}). To date, works on matrix multiplication and solving linear equations using quantum computing have addressed the cases of real- and complex-valued matrices. As far as we know, our proposal here is the first attempt to address this problem matrices over $\F$. 

In addition, $\mathrm{Mod2VQLS}$ can tackle matrices of any size and rank, whereas both HHL and VQLS are restricted to square matrices with full rank (invertible operators). Indeed the runtime of both HHL and VQLS depends on the condition number of the coefficient matrix, so effectively they assume the coefficient matrix is not only invertible, but \textit{robustly} so (well-conditioned).

The ability to work with matrices of different sizes and ranks is especially beneficial for solving problems where the matrix sizes are not fixed in advance.
A case-in-point is the integer factoring problems investigated by \cite{aboumrad2023lattice}, which motivated the development of $\mathrm{Mod2VQLS}$.
In this application, each equation (each row of the matrix $A$) is discovered independently in a massively-parallel data collection phase.
We are not looking for a unique solution to a fixed system of equations, but {\it any} solution to any subset of this system of equations.
The  $\mathrm{Mod2VQLS}$ works especially well for this, because the optimized variational circuit produces a superposition over multiple solutions to $Ax = b$.

On the flip side, while HHL and VQLS \textit{promise} an exponential speed-up over classical solvers, because they leverage a dense amplitude encoding that requires $\log_2 n$ qubits to represent the linear system $A' x' = b'$ with an $n \times n$ coefficient matrix $A'$, $\mathrm{Mod2VQLS}$ can only promise a polynomial speed-up over the fastest-known classical solvers (which are already polynomial), because $\mathrm{Mod2VQLS}$ requires circuits with $m + n$ qubits for solving the system $Ax = b$ with an $m \times n$ coefficient matrix $A$.

However, the $\mathrm{Mod2VQLS}$ encoding furnishes several advantages that address the HHL caveats delineated by \citep{aaronson2015read}. For starters, in HHL the load vector $b'$ must be loaded onto the quantum processor and this task in itself may require an exponential number of steps (with respect to the number of qubits in the circuit) \cite[Caveat 1]{aaronson2015read}. By contrast, in our setting the load vector $b$ is merely an $m$-bit string so it corresponds to an $m$-qubit computational basis state. As such, it can be efficiently loaded on the quantum computer using at most $m$ NOT-- or Pauli-$X$-- gates, although it turns out it is not even necessary to load $b$ onto the quantum processor in our current implementation.

Moreover, HHL assumes that the quantum computer can apply the unitary operator $e^{-iA' t}$ efficiently, for various values of $t$ \citep[Caveat 2]{aaronson2015read}. In our setting the closest analogue is applying the matrix-vector product operator defined below, and doing so requires precisely $N$ two-qubit gates, where $N$ is the number of non-zero entries in $A$. The complexity here indicates that $\mathrm{Mod2VQLS}$ directly benefits from the sparsity of $A$, which is useful in applied settings where one typically deals with large, (very) sparse matrices. For instance, the matrices encountered in recent large factoring calculations have millions of rows and columns, but only a few hundred non-zero entries per row.

Finally, extracting the solution vector $x'$ upon executing the HHL algorithm requires an exponential number of circuit measurements \citep[Caveat 4]{aaronson2015read}. Again, the exponential here is with respect to the number of qubits in the HHL circuit. By contrast, in our modulo $2$ setting the solution vector $x$ is an $n$-bit string, which corresponds to an $n$-qubit computational basis state, so it can be read-off from the optimized quantum ansatz using a fixed number of shots. An added benefit of our method is that the optimized quantum ansatz obtained upon running $\mathrm{Mod2VQLS}$ is a superposition over computational basis states corresponding to every possible solution to $Ax = b$, so we can effectively sample the solution set by measuring the optimized quantum state. This is quite useful in the factoring application, where we typically need to find more than one element in the kernel of the coefficient matrix in order to build a factor.

There are a few additional differences that distinguish our present context from that of Bravo-Prieto et al. \citep{bravoprieto2020variational}. 
They use a dense encoding which means that their cost function can be evaluated by executing quantum circuits with a logarithmic number of qubits. This is in contrast to our encoding, which requires a linear number of qubits. 
With the particular implementation details of the VQLS routine explained in \citep{bravoprieto2020variational} in mind, we see that our encoding has several benefits. For instance, in this context $\mathcal{A}$ is a unitary operation that can be performed using $N$ gates, where $N$ is the number of non-zero entries in the coefficient matrix. By contrast, the $A$ operator in \citep{bravoprieto2020variational} is a \textit{linear combination} of $L$ unitaries, and $L$ may be exponential in the number of qubits in their circuit. Moreover, the authors do not explain how to achieve such a decomposition in general; however, in Appendix A, they provide some hints for how to proceed when the coefficient matrix is sparse. In this section, they assume access to some oracle $\mathcal{O}_A$ that can effectively query the entries of $A$. Regardless, it's not clear how to implement this oracle and indeed the authors comment that, ``The gate complexity of both of these strongly depends on the form and precision with which the matrix elements of $A$ are specified'' \citep[Appendix A]{bravoprieto2020variational}.

Also in \citep{bravoprieto2020variational}, the cost function $\widehat{C}_G$ vanishes when the norm of $\ket{\psi}$ vanishes, making it deficient in the sense that non-solutions to $A x \equiv b \mod 2$ correspond to local extrema. By contrast, the variational state $\ket{\psi(\theta)}$ used here is normalized because $\mathcal{A} V(\theta)$ is a unitary operation. Thus there is no need to normalize the cost function, as in Equation (5) in \citep{bravoprieto2020variational}, and consider the more complicated so-called \textit{local} cost functions.

\section{Conclusion}

Variational quantum circuits provide design patterns that can be applied to a range of mathematical problems,
especially ones that can be expressed as optimization problems with respect to some cost function.
This paper demonstrated the Mod2VQLS system, which applies this design pattern to the problem of solving
linear systems modulo 2. The key ingredients are a circuit for computing the matrix multiplication, 
a variational ansatz including parameters to optimize, and a classical optimization process.
The rotation ansatz provided the most direct path to results, largely because its simple design made it amenable to analytical methods.

At the scales available to quantum computers today, the results here are potentially promising, but do not
compete with classical solvers. The goal of this research is to investigate potential advantages at larger scale:
we expect that medium-scale quantum computers will find their first regular commercial uses as part of
larger hybrid pipelines. Understanding the quantum opportunities, and in particular, their scaling properties
on real data sizes and distributions, will be crucial for guiding the choice between quantum opportunities.
The Mod2VQLS system presented here is a worked example of how such proposals can be implemented and 
investigated today.

\bibliography{ionq}

\ifarxiv
    \bibliographystyle{plainnat}
\fi

\end{document}